\documentclass[11pt]{article}

\usepackage{amssymb,amsmath,amsthm,sectsty,url}
\usepackage[letterpaper,hmargin=1.0in,vmargin=1.0in]{geometry}
\usepackage[boxed]{algorithm}
\usepackage{algpseudocode}
\usepackage{graphicx}
\usepackage{color}

\usepackage[pdftex,colorlinks,linkcolor=blue,citecolor=blue,filecolor=blue,urlcolor=blue]{hyperref}
\usepackage{cleveref}

\usepackage[margin=20pt,font=small,labelfont=bf]{caption}

\def\E{\mathop{\mathbb E}}

\newcommand{\FB}{\mathcal{B}}

\newcommand{\FD}{\mathcal{D}}

\newcommand{\FH}{\mathcal{H}}

\newcommand{\leps}{\log{\frac{1}{\varepsilon}}}
\newcommand{\ceil}[1]{\left \lceil #1 \right \rceil}

\newcommand{\paren}[1]{\left(#1\right)}
\newcommand{\parensqr}[1]{\left[#1\right]}
\newcommand{\Prp}[1]{\Pr\left[#1\right]}
\newcommand{\Ep}[1]{\E\left[#1\right]}
\newcommand{\maxparen}[1]{\max{\left\{#1\right\}}}
\newcommand{\logchoose}[2]{\log{ {#1 \choose #2}}}
\newcommand{\logp}[1]{\log{\paren{#1}}}

\newtheorem{theorem}{Theorem}[section]
\newtheorem{lemma}[theorem]{Lemma}

\newtheorem{corollary}[theorem]{Corollary}

\newtheorem{claim}[theorem]{Claim}

\newtheorem*{thm:LowerBound}{\Cref{LowerBoundNumber} (Restated)}

\begin{document}

\title{Tight Bounds for Sliding Bloom Filters\footnote{Research supported in part by a grant from the I-CORE Program of the Planning and Budgeting Committee, the Israel Science Foundation and the Citi Foundation.}}
\author{
    Moni Naor\footnotemark[\value{footnote}]\thanks{Incumbent of the Judith
        Kleeman Professorial Chair. Research supported in part by a grant from the
        Israel Science Foundation. Department of Computer Science and
                Applied Mathematics, Weizmann Institute of Science, Rehovot 76100,
                Israel. Email: \texttt{moni.naor@weizmann.ac.il}.}
  \and
    Eylon Yogev\thanks{Department of Computer Science and
        Applied Mathematics, Weizmann Institute of Science, Rehovot 76100,
        Israel. Email: \texttt{eylon.yogev@weizmann.ac.il}.}
}

\pagestyle{plain}

\maketitle

\begin{abstract}
A Bloom filter is a method for reducing the space (memory) required for representing a set by allowing a small error probability. In this paper we consider a Sliding Bloom Filter: a data structure that, given a stream of elements, supports membership queries of the set of the last $n$ elements (a sliding window), while allowing a small error probability and a slackness parameter. 

The problem of sliding Bloom filters has appeared in the literature in several communities, but this work is the first theoretical investigation of it. 

We formally define the data structure and its relevant parameters and analyze the time and memory requirements needed to achieve them. We give a low space construction that runs in $O(1)$ time per update with high probability (that is, for all sequences with high probability all operations take constant time) and provide an almost matching lower bound on the space that shows that our construction has the best possible space consumption up to an additive lower order term.
\end{abstract}

\section{Introduction}
Given a stream of elements, we consider the task of determining whether an element has appeared in the last $n$ elements of the stream. To accomplish this task, one must maintain a representation of the last $n$ elements at each step. One issue, is that the memory required to represent them might be too large and hence an approximation is used. We formally define this approximation and completely characterize the space and time complexity needed for the task.

In 1970 Bloom~\cite{Bloom70} suggested an efficient data structure, known as the `\emph{Bloom filter}', for reducing the space required for representing a set $S$ by allowing a small error probability on membership queries. The problem is also known as the approximate membership problem (however, we refer to any solution simply as a `Bloom filter'). A solution is allowed an error probability of $\varepsilon$ for elements not in $S$ (false positives), but no errors for members of $S$. In this paper, we consider the task of efficiently maintaining a Bloom filter of the last $n$ elements (called `the sliding window') of a stream of elements.

We define an $(n,m,\varepsilon)$-\emph{Sliding Bloom Filter} as the task of maintaining a Bloom filter over the last $n$ elements. The answer on these elements must always be `Yes', the $m$ elements that appear prior to them have no restrictions ($m$ is a slackness parameter) and for any other element the answers must be `Yes' with probability at most $\varepsilon$. In case $m$ is infinite, all elements prior to the current window have no restrictions. In this case we write for short $(n,\varepsilon)$-Sliding Bloom Filter.

The problem was studied in several communities and various solutions were suggested. In this paper, we focus on a theoretical analysis of the problem and provide a rigorous analysis of the space and time needed for solving the task. We construct a Sliding Bloom Filter with $O(1)$ query and update time, where the running time is worst case with high probability (see the theorems in \Cref{sec:our} for precise definitions) and has near optimal space consumption. We prove a matching space lower bound that is tight with our construction up to an additive lower order term. Roughly speaking, our main result is figuring out the first two terms of the space required by a Sliding Bloom Filter: $n\leps + n \cdot \maxparen{\log{\leps},\log{\frac{n}{m}}}$

A simple solution to the task is to partition the window into blocks of size $m$ and for each block maintain its own Bloom filter. This results in maintaining $\ceil{\frac{n}{m}+1}$ Bloom filters. To determine if an element appeared or not we query all the Bloom filters and answer `Yes' if any of them answered positively. There are immediate drawbacks of this solution, even assuming the Bloom filters are optimal in space and time:
\begin{itemize}
\item Slow query time: $\ceil{\frac{n}{m}+1}$ Bloom filter lookups.
\item High error probability: since an error can occur on each block, to achieve an effective error probability of $\varepsilon$ we need to set each Bloom filter to have error $\varepsilon'=\frac{\varepsilon m}{n+m}$, which means that the total space used has to grow (relative to a simple Bloom filter) by roughly $n \log \frac{n+m}{m}$ bits (see \Cref{related}).
\item Sub-optimal space consumption for large $m$: the first two drawbacks are acute for small $m$, but when $m$ is large, say $m=n$, then each block is large which results in a large portion of the memory being `wasted' on old elements.
\end{itemize}
We overcome all of the above drawbacks: the query time is always constant and for \emph{any} $m$ the space consumption is nearly optimal.

Sliding Bloom Filters can be used in a wide range of applications and we discuss two settings where they are applicable and have been suggested. In one setting, Bloom filters are used to quickly determine whether an element is in a local web cache \cite{FanCAB00}, instead of querying the cache which may be slow. Since the cache has limited size, it usually stores the least recently used items (LRU policy). A Sliding Bloom Filter is used to represent the last $n$ elements used and thus, maintain a representation of the cache's contents at any point in time.

Another setting consists of the task of identifying duplicates in streams. In many cases, we consider the stream to be unbounded, which makes it impractical to store the entire data set and answer queries precisely and quickly. Instead, it may suffice to find duplicates over a sliding window while allowing some errors. In this case, a Sliding Bloom Filter (with $m$ set to infinity) suffices and in fact, we completely characterize the space complexity needed for this problem.

\subsection{Problem Definition}
Given a stream of elements $\sigma=x_1,x_2,...$ from a finite universe $U$ of size $u$, parameters $n$, $m$ and $\varepsilon$, such that $n < \varepsilon u $, we want to approximately represent a sliding window of the $n$ most recent elements of the stream. An algorithm $A$ is given the elements of the stream one by one, and does not have access to previous elements that were not stored explicitly. Let $\sigma_t=x_1,\dots,x_t$ be the first $t$ elements of the stream $\sigma$ and let $\sigma_t(k)=x_{\max{(0,t-k+1)}},\dots ,x_t$ be the last $k$ elements of the stream $\sigma_t$. At any step $t$ the current window is $\sigma_t(n)$ and the $m$ elements before them are $\sigma_{t-n}(m)$. If $m=\infty$ then define $\sigma_{t-n}(m)=x_1,\dots,x_{t-n}$. Denote $A(\sigma_t,x) \in \{\mbox{`Yes'},\mbox{`No'}\}$ the result of the algorithm on input $x$ given the stream $\sigma_t$. We call $A$ an $(n,m,\varepsilon)$-\emph{Sliding Bloom Filter} if for any $t \ge 1$ the following two conditions hold:
\begin{enumerate}
\item For any $x \in \sigma_t(n)$: $\Pr[A(x)= \mbox{`Yes'}] = 1$
\item For any $x \notin \sigma_t(n+m): \Pr[A(x)= \mbox{`Yes'}] \le \varepsilon$
\end{enumerate}
where the probability is taken over the internal randomness of the algorithm $A$. Notice that for an element $x \in \sigma_{t-n}(m)$ the algorithm may answer arbitrarily (no restrictions). See Figure 1.

\begin{figure}[ht!]
\centering
\includegraphics[scale=0.35]{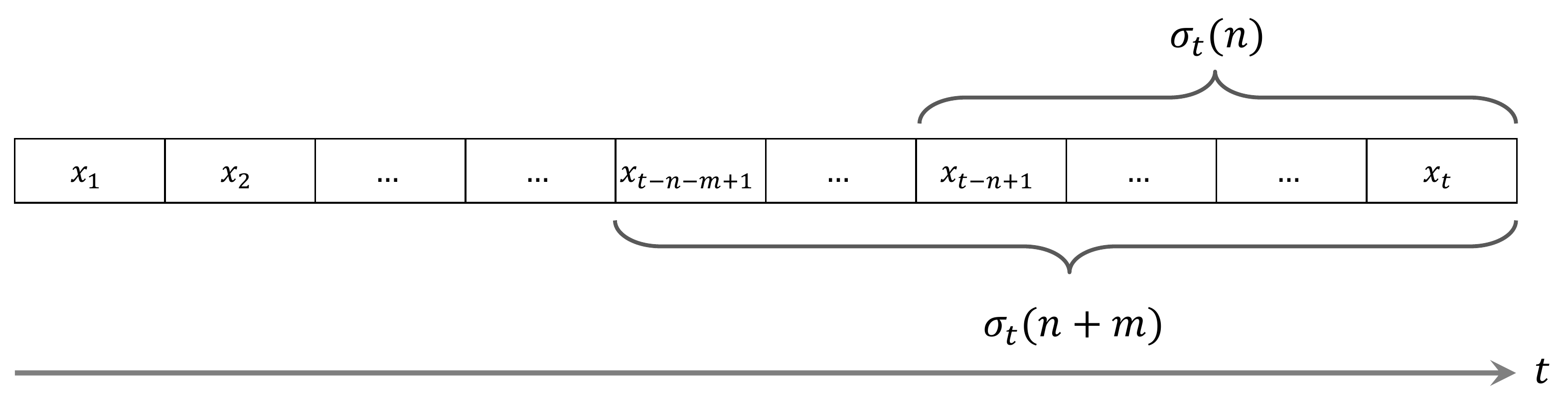}
\caption{The sliding window of the last $n$ and $n+m$ elements}
\label{overflow}
\end{figure}

An algorithm $A$ for solving the problem is measured by its memory consumption, the time it takes to process each element and answer a query. We denote by $|A|$ the maximum number of bits used by $A$ at any step. The model we consider is the unit cost RAM model in which the elements are taken from a universe of size $u$, and each element can be stored in a single word of length $w = \log{u} $ bits. Any operation in the standard instruction set can be executed in constant time on $w$-bit operands. This includes addition, subtraction, bitwise Boolean operations, left and right bit shifts by an arbitrarily number of positions, and multiplication. The unit cost RAM model is considered the standard model for the analysis of the efficiency of data structures.

An element not in $S$ on which the data structure accepts is called a false positive. At any point in time, the fraction of false positives in $U$ is called the false positive rate.

\subsection{Our Contributions}
\label{sec:our}
We provide tight upper and lower bounds to the $(n,m,\varepsilon)$-problem. In fact, we achieve space optimality up to the second term. Our first contribution is a construction of an efficient Sliding Bloom Filter: it has query time $O(1)$ worst case and update time $O(1)$ worst case with high probability, for the entire sequence. For $\varepsilon=o(1)$ the space consumption is near optimal: the two leading terms are optimal in constants.

\begin{theorem}\label{UpperBoundNumber}
For any $m>0$, and sufficiently large $n$ there exists an $(n,m,\varepsilon)$-Sliding Bloom Filter having the following space and time complexity on a unit cost RAM:

\begin{description}
\item \textbf{Time}: Query time is $O(1)$ worst case. For any polynomial $p(n)$ and sequence of at most $p(n)$ operations, with probability at least $1-1/p(n)$, over the internal randomness of the data structure, all insertions are performed in time $O(1)$ worst case.
\item
\textbf{Space}: the space consumption is: $(1+o(1))\paren{n\leps + n \cdot \maxparen{\log{\frac{n}{m}},\log{\leps}}}$.\\
In particular, for constant error $\varepsilon$ we get that the space consumption is: $n\logp{\frac{n}{m}} + O(n)$. Otherwise, for sub-constant $\varepsilon$ that satisfies $\varepsilon = 2^{-O(\log^{1/3}{n})}$ we get that:
\begin{enumerate}
\item If  $m \ge \varepsilon n$ then the space consumption is: $n\leps + n \cdot \maxparen{\log{\frac{n}{m}},\log{\leps}} + O(n)$
\item If $m < \varepsilon n$ then the space consumption is: $n\leps + (1+o(1))n\log{\frac{n}{m}}$
\end{enumerate}
\end{description}

\end{theorem}

The challenge we face is achieving constant time operations while space consumption remains very tight. In designing our algorithm we assemble ideas from several previous works along with new ones. The basic skeleton of the algorithm shares ideas with the work of Zhang and Guan~\cite{ZhangG08}, however, their algorithm is based on the traditional Bloom filter and has immediate drawbacks: running time is super-constant and the space is far from optimal. To get an error probability of $\varepsilon$ they use $M=O(n\log{n}\leps)$ bits, and moreover this is assuming the availability of truly random hash functions.

Thorup~\cite{Thorup11} considered a similar data structure of hash tables with timeouts based on linear probing. He did not allow any error probability nor any slackness (i.e. $\varepsilon=0$ and $m=0$ in our terminology). The query time, as is general for linear probing, is only constant in expectation, and the space is only optimal within a constant factor.

Pagh, Pagh and Rao~\cite{PaghPR05} showed that the traditional construction of a Bloom filter can be replaced with a construction that is based on dictionaries. The dictionary based Bloom filter has the advantage that its running time and space consumption are completely determined by the dictionary itself, and it does not assume availability of truly random functions. Given the developments in succinct dictionaries, using this alternative has become more appealing.

Our algorithm is conceptually similar to the work of Zhang and Guan. However, we replace the traditional implementation of the Bloom filter with a dictionary based one. As the underlying dictionary, we use the state of the art dictionary given by Arbitman, Naor and Segev~\cite{ArbitmanNS10}, known as Backyard Cuckoo Hashing. Then we apply a similar method of lazy deletions as used by Thorup on the Backyard Cuckoo Hashing dictionary. Moreover, we introduce a slackness parameter $m$ and instead of storing the exact index of each element we show a trade-off parameter $c$ between the accuracy of the index stored and the number of elements we store in the dictionary. Optimizing $c$ along with the combined methods described gives us the desired result: constant running time, space consumption of nearly $n\leps + n \cdot \maxparen{\log{\leps},\log{\frac{n}{m}}}$ which is optimal in both leading constants and no assumption on the availability of truly random functions. We inherit the implementation complexity of the dictionary, and given an implementation of one, it is relatively simple to complete the algorithm's implementation.

Our second contribution, and technically the more involved one, is a matching space lower bound. We prove that if $\varepsilon=o(1)$ then any Sliding Bloom Filter must use space that is within an additive low order term of the space of our construction, regardless of its running time.

\begin{theorem}\label{LowerBoundNumber}
Let $A$ be an $(n,m,\varepsilon)$-Sliding Bloom Filter where $n < \varepsilon u$, then
\begin{enumerate}
\item If $m>0$ then $|A| \ge n\leps + n \cdot \maxparen{\log{\frac{n}{m}},\log{\leps}} - O(n)$
\item If $m=\infty$ then $|A| \ge n\leps + n\log{\leps} - O(n)$
\end{enumerate}
\end{theorem}

From \Cref{UpperBoundNumber,LowerBoundNumber} we conclude that making $m$ larger than $n/\leps$ does not make sense: one gets the same result for any value in $[n/\leps, \infty)$.
When $m$ is small (less than $\varepsilon n$), then the dominant expression in both the upper and lower bounds is $n \logp{\frac{n}{m}}$.

The lower bound is proved by an encoding argument which is a common way of showing lower bounds in this area (see for example \cite{PaghSW12}). Specifically, the idea of the proof is to use $A$ to encode a set $S$ and a permutation $\pi$ on the set corresponding to the order of the elements in the set. We consider the number of steps from the point an element is inserted to $A$ to the first point where $A$ answers `No' on it, and we define $\lambda$ to be the sum of $n$ such lengths. If $\lambda$ is large, then there is a point where $A$ represents a large portion of $S$, which benefits in the encoding of $S$. If $\lambda$ is small, then $A$ can be used as an approximation of $\pi$, thus encoding $\pi$ precisely requires a small amount of bits. In either case, the encoding must be larger than the entropy lower bound\footnote{The entropy lower bound is base 2 logarithm of the size of the set of all possible inputs. In our case, all possible pairs $(S,\pi)$.} which yields a bound on the size of $A$. The optimal value of the trade-off between representing a larger set or representing a more accurate ordering is achieved by our construction. In this sense, our upper bound and lower bound match not only by `value' but also by `structure'.

\subsection{Related Work and Background}\label{related}
The data structure for the approximate set membership as suggested by Bloom in 1970 \cite{Bloom70} is relatively simple: it consists of a bit array which is initiated to `0' and $k$ random hash functions. Each element is mapped to $k$ locations in the bit array using the hash functions. To insert an element set all $k$ locations to 1. On lookup return `Yes' if all $k$ locations are 1. To achieve an error probability of $\varepsilon$ for a set of size $n$ Bloom showed that if $k=\leps$ then the length of the bit array should be roughly $1.44n\leps$ (where the 1.44 is an approximation of $\log_2(e)$). Since its introduction Bloom filters have been investigated extensively and many variants, implementations and applications have been suggested. We call any data structure that implements the approximate set membership a `Bloom filter'. A comprehensive survey (for its time) is Broder and Mitzenmacher~\cite{BroderM03}.

A lot of attention was devoted for determining the exact space and time requirements of the approximate set membership problem. Carter et al.~\cite{CarterFGMW78} proved an entropy lower bound of $n\leps$, when the universe $U$ is large. They also provided a reduction from approximate membership to {\em exact} membership, which we use in our construction. The retrieval problem associates additional data with each element of the set. In the static setting, where the elements are fixed and given in advance, Dietzfelbinger and Pagh propose a reduction from the retrieval problem to approximate membership \cite{DietzfelbingerP08}. Their construction gets arbitrarily close to the entropy lower bound.

In the dynamic case, Lovett and Porat~\cite{LovettP10} proved that the entropy lower bound cannot be achieved for any {\em constant} error rate. They show a lower bound of $C(\varepsilon)\cdot n\leps$ where $C(\varepsilon) > 1$ depends only on $\varepsilon$. Pagh, Segev and Wieder~\cite{PaghSW12} showed that if the size $n$ is not known in advance then at least $(1-o(1))n\leps + \Omega(n\log{\log{n}})$ bits of space must be used. The Sliding Bloom Filter is in particular also a Bloom Filter in a dynamic setting, thus the \cite{LovettP10} and~\cite{PaghSW12} bounds are applicable.

As discussed, Pagh, Pagh and Rao \cite{PaghPR05} suggested an alternative construction for the Bloom filter. They used the reduction of Carter et al.\ to improve the traditional Bloom filter in several ways: Lookup time becomes $O(1)$ independent of $\varepsilon$, has succinct space consumption, uses explicit hash functions and supports deletion. In the dynamic setting for a constant $\varepsilon$ we do not know what is the leading term in the memory needed, however, for any sub-constant $\varepsilon$ we know that the leading term is $n \leps$: Arbitman, Naor and Segev present a solution, called `Backyard Cuckoo Hashing', which is optimal up to an additive lower order term (i.e., it is a succinct representation) \cite{ArbitmanNS10}. Thus, in this paper we focus on sub-constant $\varepsilon$.

The model of sliding windows was first introduced by Datar et al.~\cite{DatarGIM02}. They consider maintaining an approximation of a statistic over a sliding window. They provide an efficient algorithm along with a matching lower bound.

Data structures for problems similar to the Sliding Bloom Filters have been studied in the literature quite extensively over the past years. The simple solution using $m=n$ consists of two large Bloom filters which are used alternatively. This method known as \emph{double buffering} was proposed for classifying packets caches \cite{ChangLF04}. Yoon~\cite{Yoon10} improved this method by using the two buffers simultaneously to increase the capacity of the data structure. Deng and Rafiei~\cite{DengR06} introduced the Stable Bloom filter and used it to approximately detect duplicates in stream. Instead of a bit array they use an array of counters and to insert an element they set all associated counters to the maximal value. At each step, they randomly choose counters to decrease and hence older element have higher probability of being decreased and eventually evicted over time. Metwally et al.~\cite{MetwallyAA05} showed how to use Bloom filters to identify duplicates in click streams. They considered three models: Sliding Windows, Landmark Windows and Jumping Windows and discuss their relations. A comprehensive survey including many variations is given by Tarkoma et al.~\cite{TarkomaRL12}. However, as far as we can tell, no formal definition of a Sliding Bloom Filter as well as a rigorous analysis of its space and time complexity, appeared before.

\section{The Construction of a Succinct Sliding Bloom Filter}\label{sec:contruction}

Our algorithm uses a combination of transforming the approximate membership problem to the exact membership problem plus a solution to the retrieval problem. On an input $x$, we store $h(x)$, for some hash function $h$, in a dynamic dictionary and in addition store some information on the last time where $x$ appeared. We consider the stream to be divided into generations of size $n/c$ each, where $c$ is a parameter that will be optimized later. The first $n/c$ elements are generation 1, the next $n/c$ elements are generation 2 etc. The current window contains the last $n$ elements and consists of at most $c+1$ different generations. Therefore, at each step, we maintain a set $S$ that represents the last $c+1$ generations (that is, at most $n+n/c$ elements) and count the generations mod $(c+1)$. In addition to storing $h(x)$, we associate $s=\logp{c+1}$ bits indicating the generation of $x$. Every $n/c$ steps, we delete elements associated with the oldest generation. We adjust $c$ to optimize the space consumption while requiring $n/c \le m$.

In this section, we describe the algorithm in more detail. We first present the transformation from approximate to exact membership (Section 2.1). We define a dynamic dictionary and the properties we need from it in order to implement our algorithm (Section 2.3). Then, we describe the algorithm in two stages, using any dictionary as a black box. The memory consumption is merely the memory of the dictionary and therefore we use one with succinct representation. At first, in Section 2.3, the running time will not be optimal and depend on $c$ (which is not a constant), even if we use optimal dictionaries. Then, in Section 2.4, we describe how to eliminate the dependency on $c$ as well as deamortizing the algorithm, making the running time constant for each operation. This includes augmenting the dictionary, and thus it can no longer be treated as a black box. We prove correctness and analyze the resulting memory consumption and running time.

\subsection{Approximate Membership and Exact Membership}\label{sec:approxToExact}
Carter et al.~\cite{CarterFGMW78} showed a transformation from approximate membership to exact membership that works as follows. We want to represent a set $S$ of size $n$ and support membership queries in the following manner: For a query on $x \in S$ we answer `Yes' and for $x \notin S$ we answer `Yes' with probability at most $\varepsilon$. Choose a hash function $h \in \FH$ from a universal family of hash functions mapping $U \rightarrow \left [n /\varepsilon \right]$. Then for any $S$ of size at most $n$ it holds that for any $x \in U$:
$$
\Pr_h[h(x) \in h(S)] \le \sum\limits_{y \in S}\Pr_h[h(x)=h(y)] \le n\frac{\varepsilon}{n}=\varepsilon
$$
where the first inequality comes from a union bound and the second from the definition of a universal hash family. This implies that storing $h(S)$ suffices for solving the approximate membership problem. This dictionary-based construction and the traditional construction can be viewed as lying on a spectrum - the former writes many bits in one location, whereas the latter writes one bit in many locations.

To store $h(S)$ we use an exact dictionary $\FD$, which supports \texttt{insert} (including associated data), \texttt{delete} and \texttt{update} procedures (the update procedure can be simulated by a delete followed by an insert). While most dictionaries support these basic procedures, we require $\FD$ to additionally support the ability of \emph{scanning}. We further discuss these properties in the next section.\\

\noindent
\textbf{Number of false positives}:
We note that in addition to the error bound on each element, we can bound the total number of false positives in the universe. Any hash family $\FH$ divides the universe to $\ceil{n/\varepsilon}$ `bins', and the number of false positives is the total number of elements in any bin containing an element from $S$. If $\FH$ divides $U$ to (roughly) equally sized bins, each of size at most $\ceil{\varepsilon u/n}$, then the total number of false positives is at most $n \cdot \ceil{\varepsilon u/n} \le \varepsilon u + n$. A simple example of such a hash family can be obtained by choosing a prime $p \ge u$ then defining $\FH$ to be $h_{a}(x)=((ax \mod{p}) \mod{\ceil{n/\varepsilon}})$, where $a$ is a random integer modulo $p$ with $a \ne 0$ \cite{CarterW79}. In this case, the bound holds with certainty for \emph{any} function $h \in \FH$. This property is not guaranteed by the traditional construction of Bloom, and we further discuss it in \Cref{sec:lowebound}.

\subsection{Succinct Dynamic Dictionary}
The information-theoretic lower bound on the minimum number of bits needed to represent a set $S$ of size $n$ out of $M$ different elements is $\FB=\FB(M,n) = \ceil{\log{{M \choose n}}} = n\log{M} - n\log{n} + O(n)$. A succinct representation is one that uses $(1+o(1))\FB$ bits \cite{Dem07}. A significant amount of work was devoted for constructing dynamic dictionaries over the years and most of them are appropriate for our construction. Some have good theoretical results and some emphasize the actual implementation. In order for the reduction to compete with the Bloom filter construction (in terms of memory consumption) we must use a dynamic dictionary with succinct representation. There are several different definitions in the literature for a \emph{dynamic} dictionary. A static dictionary is a data structure storing a finite subset of a universe $U$, supporting only the \texttt{member} operation. In this paper, we refer to a dynamic dictionary where only an upper bound $n$ on the size of $S$ is given in advance and it supports the procedures \texttt{member}, \texttt{insert} and \texttt{delete}. The memory of the dictionary is measured with respect to the bound $n$.

In addition to storing $h(S)$, we assume $\FD$ supports associating data with each element. Specifically, we want to store $s$-bits of data with each element, where $s$ is fixed and known in advance. Finally, we assume the dictionary supports \emph{scanning}, that is, the ability to go over the associated data of all elements of the dictionary, and delete the element if needed. Using the scanning process, we scan the generations stored in the dictionary and delete elements of specific generations.

Several dynamic dictionaries can be used in our construction of a Sliding Bloom Filter. The running time and space consumption are directly inherited from the dictionary, making it an important choice. We use the `Backyard Cuckoo Hashing' construction of \cite{ArbitmanNS10} (but other alternative are possible). It supports \texttt{insert} and \texttt{delete} in $O(1)$ worst case with high probability while having a succinct representation. Implicitly in their work, they support associating any fixed number of bits and scanning. When $s$-bits of data are associated with each $x \in S$, the representation lower bound becomes $\FB + ns$ bits. For concreteness, the memory consumption of their dictionary is $(1+o(1))\paren{\FB + ns}$, where the $o(1)$ hides the expression $\frac{\log{\log{n}}}{\log^{1/3}{n}}$.

\subsection{An Algorithm with Dependency on $\varepsilon$}\label{Algorithm1}
Initiate a dynamic dictionary $\FD$ of size $n'=n\paren{1+\frac{1}{c}}$ as described above. Let $\FH = \{h:U \rightarrow [n'/\varepsilon]\}$ be a family of universal hash functions and pick $h \in \FH$ at random. At each step maintain a counter $\ell$ indicating the current generation and a counter $i$ indicating the current element in the generation. At every step $i$ is increased and every $n/c$ steps $i$ is reset back to 0 and $\ell$ is increased mod $(c+1)$.

To insert an element $x$ check if $h(x)$ exists in $\FD$. If not then insert $\langle h(x),\ell \rangle$ (insert $h(x)$ associated with $\ell$) into $\FD$. If $h(x)$ is in $\FD$, then update the associated data of $h(x)$ to $\ell$. Finally, update the counters $i$ and $\ell$. If $\ell$ has increased (which happens every $n/c$ steps) then \emph{scan} $\FD$ and delete all elements with associated data equal to the new value of $\ell$.

To query the data structure on an element $x$, return whether $h(x)$ is in $\FD$. See Algorithm 1 for pseudo-code of the insert and lookup procedures.

\begin{algorithm}
\caption {Pseudo-code of the Insert and Lookup procedures}
\texttt{Insert($x$):}
\begin{algorithmic}[1]

\If {$h(x)$ is a member of $\FD$}
	\State update $h(x)$ to have data $\ell$
\Else
	\State insert $\langle h(x),\ell \rangle$ into $\FD$
\EndIf
\State maintain counters $i$ and $\ell$
\If {the value of $\ell$ has changed}
	\State scan $\FD$ and delete elements of generation $\ell$
\EndIf
\newline
\end{algorithmic}

\texttt{Lookup($x$):}
\begin{algorithmic}[1]
\Procedure {member}{$x$}
\If {$h(x)$ is a member of $\FD$}
	\State return `Yes'
\Else
	\State return `No'
\EndIf
\EndProcedure
\end{algorithmic}
\end{algorithm}

\noindent
\textbf{Correctness}:
We first notice that $\FD$ is used correctly and never represents a set of size larger than $n'$. In each step we either insert an element to generation $\ell$ or move an existing element to generation $\ell$. In any case, each generation consists of at most $n/c$ elements in $\FD$. Each $n/c$ we evict a whole generation, assuring no more than $c+1$ generations are present in the dictionary at once. Thus, at most $n'$ elements are represented at any given step.

Next we prove that for any time $t$ the three conditions in \Cref{UpperBoundNumber} hold. The first condition follows directly from the algorithm. Assume $h(x)$ is inserted with associated generation $\ell=j$. Notice that its associated generation can only increase. $h(x)$ will be deleted only when $\ell$ completes a full cycle and its value is $j$ again, which takes at least $n$ steps. Thus, for any $x \in \sigma_t(n)$, $h(x)$ is in $\FD$ and the algorithm will always answer `Yes'.

For the second condition assume that $x \notin \sigma_t(n+m)$ and notice that $n+m$ is at least $c+1$ generations. Assume w.l.o.g.\ that $S=\{y_1,\dots ,y_{n'}\}$ ($S$ could have less than $n'$ elements) is the set of elements represented in $\FD$ at time $t$. Then $\Pr[h(x)=y_i]=\frac{\varepsilon}{n'} $ for all $i \in [n']$. Therefore, using a union bound we get that the total false positive probability is
$$
\Pr[A(x) = \mbox{`Yes'}] =
\Pr[h(x) \in h(S)] \le
\sum\limits_{i=1}^{n'}\Pr[h(x) = y_i] \le
\varepsilon
$$

\noindent
\textbf{Memory consumption:}
The bulk of memory is used for storing $\FD$. In addition, we need to store two counters $i$ and $\ell$ and the hash function $h$, which together take $O(\log{n})$ bits. $\FD$ stores $n'$ elements out of $M=[n'/\varepsilon]$ while associating each with $s=\log{c}$ bits. Using the `Backyard Cuckoo Hashing' dictionary yields a total space of
\begin{eqnarray*}
	\paren{1+o(1)}\paren{\FB\paren{\frac{n'}{\varepsilon}, n'} + n's} & = & \paren{1+o(1)} \cdot n \paren{1+\frac{1}{c}} \paren{ \leps + \log{c} + 1}
\end{eqnarray*}
We minimize this expression, as a function of $c$, and get that the minimum is at the solution to $c-\log{c}=\leps-1$. An approximate solution is $c=\leps$ and requiring that $n/c \le m$ yields that $c=\maxparen{\leps,m/n}$ and the total space is
$$
\paren{1+o(1)}\paren{n\leps + n \cdot \maxparen{\log{\leps},\log{\frac{n}{m}}}}
$$
as required. As mentioned, the $o(1)$ hides the term $\frac{\log{\log{n}}}{\log^{1/3}{n}}$, therefore if $\varepsilon = 2^{-{O(\frac{\log^{1/3}{n}}{\log{\log{n}}})}}$ then the product of $o(1)$ with $n \leps$ is $O(1)$. If $m \ge \varepsilon n$ then the product of $o(1)$ with $n \maxparen{\log{\leps},\logp{n/m}}$ is $O(1)$ as well. Thus, we can write the space consumption as:
$$n\leps + n \cdot \maxparen{\log{\leps},\log{\frac{n}{m}}} + O(n).$$
Otherwise, if $m < \varepsilon n$ then we can write it as:
$$n\leps + (1+o(1))n\log{\frac{n}{m}}.$$
If $\varepsilon=O(1)$ then $n\leps=O(n)$ and $n\log{\leps}=O(n)$ and we can write it as:
$$ n\log{\frac{n}{m}} + O(n).$$
\noindent
\textbf{Running time:}
Assume that $\FD$ supports $O(1)$ running time worst case for all procedures. The lookup procedure performs a single query to $\FD$ and hence always runs in $O(1)$. In the insert procedure, every $n/c$ steps, the value of $\ell$ is updated and we scan all elements in $\FD$ deleting old elements. For any other step, the running time is $O(1)$. Therefore, the total running time for $n/c$ steps is $O(n')$, which is $O(c)$ amortized running time. If $m \ge \leps$ then $c=\leps$ and the running time is $\leps$, otherwise it is $O\paren{\frac{n}{m}}$, which in both cases is not constant. We now show how to eliminate the large step, making the running time $O(1)$ worst case. Using the `Backyard Cuckoo Hashing' dictionary we get that the total running time including the dictionary's operations is $O(1)$ worst case with high probability (over internal randomness of the dictionary).

\subsection{Reducing the Running Time to Constant}

The main load of the algorithm of \Cref{Algorithm1} stems from the need to scan the entire dictionary to delete old elements. The issue we have to deal with in order to reduce the time of each step to be constant (independent of $\varepsilon$) is that the scanning operation is done too many times: each $n/c$ steps we scan the dictionary which is a total of $O(cn)$ operations over $n$ steps. Another drawback of that algorithm is that the scanning process is all done in one step, hence we can get only an amortized result. We modify the algorithm to solve these two issues simultaneously: only one scanning is performed every $n'$ steps and the scanning processes is deamortized and spread over these $n'$ steps.

The first modification is to extend the range of the generations counter $\ell$ to loop between $0$ and $2c+2$ (instead of between 0 and $c+1$). This lets us distinguish between elements of the last $2c+2$ generation and enables a window of size $n'$ to delete old elements before the counter overrides them with a new generation. At any moment, only the $c+1$ recent generations are considered active and the rest slated to be deleted.

The second modification is to combine many scanning processes to one, which is spread over a sequence of $n'$ steps. The scanning process needs to support running in small steps while allowing other operations to run concurrently. We should be able to save its state, then allow other operations to run and finally resume its state and continue the scanning process. Instead of scanning all the $n'$ elements in one step, we scan two elements at each step and save the scanning index such that we are able to continue from that point. Thus, after $n'/2$ steps all $n'$ elements of the dictionary are scanned.

These modifications raise two new problems with which we need to deal. First, the dictionary is initialized to be of size $n'$ and since we do not delete old elements immediately, there might be more than $n'$ elements present in the dictionary. Notice, however, that the number of \emph{active} elements present will never exceed $n'$. Second, since the scanning is done in small steps concurrently with other operations, it might miss elements the have been moved by other operations. In any case the dictionary has been modified, the scanning process should succeed in scanning all elements nevertheless.

To solve this, we need the dictionary to be able to consider non-active elements as deleted such that they do not interfere with other operations: whenever a non-active element is encountered it is simply deleted. Supporting this requires some additional properties from the dictionary. Later, for concreteness, we describe how to modify the `Backyard Cuckoo Hashing' dictionary to support these properties.

It is not clear whether all dictionaries can be modified to support this property, since the dictionary might have some implicit representation of various elements using the same memory space. However, the property can be supported assuming each element has a unique memory space in which it is  represented, called a `cell'; we do not assume that the dictionary is `systematic', i.e. that the string encodes the element directly, but rather that as in `traditional' hash tables the content of the cell plus its location and some other easily accessible information determine the element uniquely. We assume that given a cell, we can figure out the associated data with the element of the cell and delete the element of this cell from the dictionary. An insert or delete procedure may modify a constant number of cells. Elements of cells which were accessed are called the accessed elements. We assume the cells have some order in which we can scan them and save an index indicating the state of the scanning process using $o(n)$ bits of memory (actually it is $O(\log{n})$).

Assuming the dictionary supports these properties, we can modify its lookup and insert procedures to check whether any accessed element needs to be deleted. For example, an insert procedure may move an element from one cell to another, which was already scanned. Thus, before moving or changing a cell we scan it and delete it if it's old. This way, each element is scanned either by the scanning process or by an insert or lookup procedure. Moreover, we change the Lookup procedure to return `Yes' on input $x$ only if $h(x)$ exists in $\FD$ \textbf{and} its associated generation is active.

A cell occupied by an old element, will be deleted whenever accessed, thus effectively not occupying space in dictionary. Since elements might be deleted only after $n'/2$ steps, it could be the case that more than $n'$ elements are present in the dictionary. However, this way, the old elements do not interfere: when an old item is encountered during the insertion it is deleted, as described above. Hence, effectively when an item is inserted the data structure has at most $n'$ elements and it will have a valid place.

We discuss implementing these requirements in the `Backyard Cuckoo Hashing' construction (see pseudo-code in Figure 2 of their paper). Their hashing scheme is based on two-level hashing, the first level consists of an array $T_0$ of bins of size $d$ and the second level consists of Cuckoo hashing which includes two arrays, $T_1$ and $T_2$ and a queue, $Q$. The cells are the $d$ cells in each bin of $T_0$, the cells of $T_1$, $T_2$ and $Q$. Each element is implicitly stored in a unique cell in one of the components.

Scanning the cells is achieved by going over the cells of each component and saving an index of the current component and cell within the component. The lookup and delete procedures are simple and does not involve moving cells. The insert procedure is more involved and may move cells from one component to another, e.g. a cell from $Q$ might be moved to $T_0$. Since the running time is constant, so is the number of accessed elements. The procedure can be easily modified such that before \emph{any} cell is accessed it is first scanned, and deleted if old. If there are less than $n'$ \emph{active} elements in the dictionary, then an insert operation will succeed, removing old elements if required. After these modifications, the `Backyard Cuckoo Hashing' dictionary supports all needed requirements for our construction of an $(n,\varepsilon)$-Sliding Bloom Filter.

We analyze the running time of the modified construction. At each step, we scan two elements and delete them if necessary. The delete operation always takes constant time. The insert procedure was modified to delete old accessed elements when encountered. Since the insert operation takes constant time in worst case with high probability, then with the same probability, it will access only a constant number of cells. Hence, deleting accessed elements will increase the running time, but it will remain a constant. Similarly, the modified lookup procedure also remains constant. Overall, all operations remain constant in the worst case, where the insert operation has constant running time, with high probability. This completes the proof of \Cref{UpperBoundNumber}.

\section{A Tight Space Lower Bound}\label{sec:lowebound}
In this section we present a matching space lower bound to our construction. For simplicity, we first introduce what we call the `absolute false positive assumption'. We define it and use it in the proof of \Cref{sec:lower}, and in \Cref{RemoveAssumption} we show how to get the same lower bound without it.

Recall that at any point in time, the false positive rate is the fraction of false positive elements. According to the definition of a Sliding Bloom Filter, we are not assured that there are no `bad' points in time where the false positive rate is much higher than its expectation, and in fact it could get as high as $1$.

We call the property that throughout the lifetime of the data structure at {\em all} points in time the false positive rate is at most $\varepsilon$ the \emph{absolute false positive assumption}. This assumption is a desirable property from a Sliding Bloom Filter and reasonable constructions, including ours\footnote{See discussion at \Cref{sec:approxToExact}} in \Cref{sec:contruction}, enjoy it.

An (artificial) example of a Sliding Bloom Filter for which the assumption does not hold can be obtained by taking any $(n,\varepsilon)$-Sliding Bloom Filter and modifying it such that it chooses a random index $k \in [1,n]$ and at step $k$ of the stream it always answers `Yes'. This results in an $(n,\varepsilon + \frac{1}{n})$-Sliding Bloom Filter in which there will \emph{always} be some point at which the false positive rate is high.

\subsection{Proof Under the Absolute False Positive Assumption}\label{sec:lower}

\begin{theorem}\label{LoweBoundAssumptionNumber}
Let $A$ be an $(n,m,\varepsilon)$-Sliding Bloom Filter where $n < \varepsilon u$. If for any stream $\sigma$ it holds that
$$\Pr[\exists i \le 3n: |\{x \in U: A(\sigma_i,x)=`Yes'\}| \ge n + 2\varepsilon u] \le \frac{1}{2}$$
then
\begin{enumerate}
\item If $m>0$ then $|A| \ge n\leps + n \cdot \maxparen{\log{\frac{n}{m}},\log{\leps}} - O(n)$
\item If $m=\infty$ then $|A| \ge n\leps + n\log{\leps} - O(n)$
\end{enumerate}
\end{theorem}

\begin{proof}
Let $A$ be an algorithm satisfying the requirements in the statement of the theorem. The main idea of the proof is to use $A$ to encode and decode a set $S \subset U$ and a permutation $\pi$ on the set (i.e. an ordered set). Giving $S$ to $A$ as a stream, ordered by $\pi$, creates an encoding of an approximation of $S$ and $\pi$: $S$ is approximated by the set of all the elements for which $A$ answers `Yes' (denoted by $\mu_A(S)$), and $\pi$ is approximated by the number of elements needed to be added to the stream in order for $A$ to "release" each of the elements in $S$ (that is, to answer `No' on it). Then, to get an exact encoding, we encode only the elements of $S$ from within the set $\mu_A(S)$. To get an exact encoding of $\pi$ we encode only the difference between the location $i$ of each element and the actual location it has been released. The key is to find the point where $A$ best approximates $S$ and $\pi$ \emph{simultaneously}.

Denote by $A_r$ the algorithm with fixed random string $r$ and let $\mu_{A_r}(\sigma) = \{x:A_r(\sigma, x)=\mbox{`Yes'}\}$. We show that w.l.o.g.\ we can consider $A$ to be deterministic. Let $V=\{\sigma:|\sigma|=2n\}$ be the set of all sequences of $2n$ \emph{distinct} elements, and let $V(r) \subseteq V$ be the subset of inputs such that $|\mu_{A_{r}}(\sigma_i)| \le n + 2\varepsilon u$ for all $1 \le i \le 3n$. Since we assumed that for any $\sigma$ we have that $\Pr_r[\exists i\le 3n: |\mu_{A_r}(\sigma_i)| \ge n + 2\varepsilon u] \le \frac{1}{2}$ then there must exist an $r^*$ such that $|V(r^*)| \ge |V|/2$. Thus, we can assume that $A$ is deterministic and encode only sequences from $V(r^*)$. Then the encoding lower bound changes from $\log{|V|}$ to $\logp{|V|/2}=\log{|V|}-1$. This loss of 1 bit is captured by the lower order term $O(n)$ in the lower bound, and hence can be ignored.

Notice that $r^*$ need not be explicitly specified in the encoding since the decoder can compute it using the description of the algorithm $A$ (which may be part of its fixed program). From now on, we assume that $A$ is deterministic (and remove the $A_r$ notation) and assume that for any $\sigma \in V(r^*)$ we have that $\mu_A(\sigma) \le n+ 2\varepsilon u \le 3\varepsilon u$.

We now make an important definition:
$$
\ell(\sigma,x) = \min{\{ \arg\min_k{\{\exists y_1,\dots,y_k \in U:A(\sigma y_1 \cdots y_k,x)=0\}}, n, m \} }
$$
$\ell(\sigma,x)$ is the minimum number of elements needed to be added to $\sigma$ such that $A$ answers `No' on $x$. Notice that $\ell(\sigma,\cdot)$ can be computed for any set $S$ given the representation of $A(\sigma)$.

We encode any set $S$ of size $2n$ and a permutation $\pi:[2n] \rightarrow [2n]$ using $A$. After encoding $S$ we compare the encoding length to the entropy lower bound of $\FB(u,2n) + \log{((2n)!)}$. Consider applying $\pi$ on (some canonical order of) the elements of $S$ and let $x_1,\dots,x_{2n}$ be the resulting elements of $S$ ordered by $\pi$. For any $i > 2n$ let $x_{i} = x_{i-2n}$, then for any $k \ge 1$ define the sequence $\sigma_k=x_1,\dots,x_k$. Let $\phi(\sigma_k) = \mu(\sigma_k) \cap S$ and define $$\Delta(\sigma_k,i)=\ell(\sigma_k,x_i) + (k-n) - i$$
Notice that, given $A(\sigma_k)$, $\Delta(\sigma_k,i)$, $k$ and $n$ one can compute the position $i$ of the element $x_i$. Define

$$
\lambda_k = \sum\limits_{i=k-n+1}^{k}\Delta(\sigma_k,i) \mbox{, and } \lambda = \max_{n \le k \le n}\lambda_k
$$

If $m \ge n$ (or $m=\infty$) then $0 \le \lambda \le n^2$, otherwise $0 \le \lambda \le nm$
\begin{lemma}\label{foranyk}
Let $k \in [n, 2n]$ then
$$
\sum_{j=k}^{k+n-1}|\phi(\sigma_j)| \ge n^2+\lambda_k.
$$
\end{lemma}
\begin{proof}
Instead of summing over $\phi(\sigma_j)$, we sum over $x_i$ and count the number of $\phi(\sigma_j)$ such that $x_i \in \phi(\sigma_j)$. For $k-n+1 \le i \le k$ we know that $x_i \in \sigma_k(n)$ and by the definition of $\ell(\sigma_k,x_i)$ we get that $x_i \in \phi(\sigma_k), \dots , \phi(\sigma_{k+\ell(\sigma_k,x_i)-1})$. For $k+1 \le i \le k+n-1$ we know that $x_i \in \phi(\sigma_i),\dots,\phi(\sigma_{k+n-1})$. Therefore:

\begin{align*}
   \sum_{j=k}^{k+n-1}|\phi(\sigma_j)| \ge & \sum\limits_{i=k-n+1}^{k}\ell(\sigma_k,x_i) + \sum\limits_{i=k+1}^{k+n-1}\paren{k+n-i} \\
	= & \sum\limits_{i=k-n+1}^{k}\ell(\sigma_k,x_i) + \frac{n(n-1)}{2} \\
   = & \sum\limits_{i=k-n+1}^k \parensqr{\ell(\sigma_k,x_i) + k - n -i} + n^2 \\
   = & \sum\limits_{i=k-n+1}^k\Delta(\sigma_k, i) + n^2 = \lambda_k + n^2
\end{align*}
\end{proof}

By averaging, we get that for any $k$ there exist some $j \in [k,k+n-1]$ such that $|\phi(\sigma_j)| \ge n+\frac{\lambda_j}{n}$. Let $k^*$ be such that $\lambda=\lambda_{k^*}$, then we know that there exist some $j^* \in [k^*,k^*+n-1]$ such that $|\phi(\sigma_j)| \ge n+\frac{\lambda}{n}$. Note that $j^*$ satisfies $n \le j \le k^*+n-1 \le 3n$ which is in the range of indices of the false positive assumption.

We include the memory representation of $A(\sigma_{j^*})$ in the encoding. The decoder uses this to compute the set $\mu(\sigma_{j^*})$, which by the absolute false positive definition we know that $|\mu(\sigma_{j^*})| \le 3\varepsilon u$. Since $|\phi(\sigma_{j^*})| \ge n+\frac{\lambda}{n}$, we need only $\FB(3\varepsilon u,n+\frac{\lambda}{n})$ bits to encode $n+\frac{\lambda}{n}$ elements of $S$ out of them. The remaining $n-\frac{\lambda}{n}$ elements are encoded explicitly using $\FB(u,n-\frac{\lambda}{n})$ bits. This completes the encoding of $S$.

To encode $\pi$ we need the decoder to be able to extract $i$ for each $x_i$. For any $x_i \in \sigma_{j^*}(n)$ the decoder uses $A(\sigma_{j^*})$ and computes $\ell(\sigma_{j^*}, x_i)$. Now, in order for the decoder to exactly decode $i$ we need to encode all the $\Delta(\sigma_{j^*}, _i)$'s. Since $\sum\limits_{i=j^*-n+1}^{j^*}\Delta(\sigma_{j^*}, _i) = \lambda_{j^*} \le \lambda$ we can encode all the $\Delta(\sigma_{j^*}, _i)$'s using $\logchoose{n+\lambda}{n}$ bits (balls and sticks method), and the remaining elements' positions will be explicitly encoded using $n\log{n}$ bits. Denote by $|A|$ the number of bits used by the algorithm $A$. Comparing the encoding length to the entropy lower bound we get

$$
|A| + \logchoose{3\varepsilon u}{n + \frac{\lambda}{n}} + \logchoose{u}{n - \frac{\lambda}{n}} + \logchoose{\lambda + n}{n} + n\log{n} \ge \logchoose{u}{2n} + \log{((2n)!)}
$$
and therefore
$$
|A| \ge (n + \frac{\lambda}{n})\leps + (n + \frac{\lambda}{n})\log{n} + (n - \frac{\lambda}{n})\log{(n - \frac{\lambda}{n})} - n\logp{\lambda+n} - O(n)
$$
Consider two possible cases for $\lambda$. If $\lambda \le 0.9n^2$ then we get
$$
|A| \ge (n + \frac{\lambda}{n})\leps + 2n\log{n} - n\logp{\lambda+n} - O(n)
$$
The minimum of this expression, as a function of $\lambda$, is achieved at $\lambda=\frac{n^2}{\leps}-n$. If $m \ge \frac{n}{\leps} - 1$ then the minimum can be achieved and we get that
$$|A| \ge n\leps + n\log{\leps} - O(n) .$$
Otherwise, if $m < \frac{n}{\leps} - 1$ then $\lambda \le mn \le \frac{n^2}{\leps}-n$ and minimum value will be achieved at $\lambda=nm$ which yields the required lower bound:
$$|A| \ge n\leps + n\log{\frac{n}{m}} - O(n) .$$
If $0.9n^2 < \lambda \le n^2$ then $m \ge 0.9n \ge \frac{n}{\leps}$. Thus, we get that
\begin{eqnarray*}
|A| & \ge & (n + \frac{\lambda}{n})\leps - (n - \frac{\lambda}{n})\log{n} + (n - \frac{\lambda}{n})\log{(n - \frac{\lambda}{n})} - O(n)
\end{eqnarray*}

the minimum of this expression, as a function of $\lambda$ between the given range is achieved at $\lambda=0.9n^2$ which yields
$$|A| \ge n\leps + n\log{\leps} - O(n)$$
as required.
\end{proof}

In the proof, we encoded a sequence of length $2n$ and we assumed that the false positive assumption holds for any such sequence. However, the only property used was the number of bits required for encoding any possible sequence. Since the lower bound includes a $O(n)$ term, we conclude that the theorem holds even for smaller sets of sequences resulting in a larger constant hidden in the $O(n)$ term. In particular, we get the following corollary, which we use to prove \Cref{LowerBoundNumber}:
\begin{corollary}\label{corollary1}
Let $W$ be a subset of sequences of length $2n$ such that $\log{|W|} = 2n\log{u} - O(n)$. Let $A$ be an $(n,m,\varepsilon)$-Sliding Bloom Filter where $n < \varepsilon u$. If for any $\sigma \in W$ it holds that
$$\Pr[\exists i \le 3n: |\{x \in U: A(\sigma_i,x)=`Yes'\}| \ge n + 2\varepsilon u] \le \frac{1}{2}$$
then
\begin{enumerate}
\item If $m>0$ then $|A| \ge n\leps + n \cdot \maxparen{\log{\frac{n}{m}},\log{\leps}} - O(n)$
\item If $m=\infty$ then $|A| \ge n\leps + n\log{\leps} - O(n)$
\end{enumerate}
\end{corollary}

\subsection{Removing the Absolute False Positive Assumption}\label{RemoveAssumption}
We show how we can remove the `absolute false positive' assumption while maintaining the same lower bound as in the original theorem. Towards this end, we construct a new data structure $A'$ which uses multiple instances of $A$. The new data structure $A'$ will work only on a specific subset of all inputs, however, we show that the number of such inputs is approximately the same and hence the same entropy lower bound holds, up to a larger constant hidden in the $O(n)$ term. Moreover, we show that on these inputs, the absolute false positive assumption holds and thus we can apply \Cref{corollary1} on $A'$. We restate and prove the theorem.

\begin{thm:LowerBound}
Let $A$ be an $(n,m,\varepsilon)$-Sliding Bloom Filter where $n < \varepsilon u$, then
\begin{enumerate}
\item If $m>0$ then $|A| \ge n\leps + n \cdot \maxparen{\log{\frac{n}{m}},\log{\leps}} - O(n)$
\item If $m=\infty$ then $|A| \ge n\leps + n\log{\leps} - O(n)$
\end{enumerate}
\end{thm:LowerBound}
\begin{proof}
In order to prove the result we need to reduce the probability of having many false positives to roughly $1/n$. To obtain this sort of bound, we partition the sequence into several subsequences on which we apply the original Sliding Bloom Filter independently. The motivation of using multiple instances of $A$ is to introduce independence between different sets of inputs.

Let $U'$ be a universe composed of $w$ copies of $U$, where $w$ will be determined later. We denote each copy by $U_i$ and we call it a world. Each world is of size $u$ and $U'$ is of size $u'=wu$. We consider only sequences such that each chunk of $w$ elements in the sequence contain exactly one element from each world. These are the only sequences the are valid for $A'$, and we denote them by $W$. Let $A_1,\dots,A_w$ be $w$ independent instances of the algorithm $A$ with parameters $(n,m,\varepsilon)$.

$A'$ works by delegating each input to the corresponding instance of $A$. On input $x \in U_i$ we insert $x$ into $A_i$, and on query $x \in U_i$ we query $A_i$ and return its answer. 
\begin{claim}
Let $n'=(n-1)w$ and $m'=(m+2)w$. Then, $A'$ is an $(n',m',\varepsilon)$-Sliding Bloom Filter for any sequence $\sigma \in W$.
\end{claim}
\begin{proof}
We show that the two properties hold for any time $t$ in the sequence. Let $\sigma \in W$, let $x \in U'$ such that $x \in U_i$, and let $\sigma^i$ be the sequence $\sigma$ limited to elements in $U_i$. $A'$ answers by $A_i$ and thus $\Pr[A'(\sigma, x)=1]=\Pr[A_i(\sigma^i, x)=1]$. We analyze $\Pr[A_i(\sigma^i, x)=1]$ in the different cases.

Since each chuck of $w$ elements contain one element from each world, each data structure will contain at most $n'/w + 1=n$ elements of the current window. Moreover, each element of the window will be present in one of the $A_i$'s. If $x \in \sigma_t(n')$ is in the current window, then since each $A_i$ has no false negatives we have $\Pr[A_i(\sigma^i,x)= \mbox{`Yes'}] = 1$.

Now suppose $x \notin \sigma_t(n'+m')$ is not in the current window and not in the $m'$ element beforehand. If $x \notin \sigma_t$ then, by the false positive probability of $A_i$, we have that $\Pr[A_i(\sigma^i, x)= \mbox{`Yes'}] \le \varepsilon$. Otherwise, $x \in \sigma_t$ but $x \notin \sigma_t(n'+m')$, and therefore at least $n'+m'$ elements have arrived after $x$. Thus, each $A_i$ has received at least $\frac{n'+m'}{w}-1=n+m$ elements, and so $x \notin \sigma^i_t(n+m)$. Since each $A_i$ is an $(n,m,\varepsilon)$-Sliding Bloom Filter we have that $\Pr[A_i(\sigma_i, x)= \mbox{`Yes'}] \le \varepsilon$
\end{proof}

We have shown that $A'$ satisfies that properties of an $(n',m',\varepsilon)$-Sliding Bloom Filter for sequences of $W$. Now, we show that the false positives assumption holds for $A'$ under sequences of $W$.

For any sequence $\sigma \in W$, for all $1 \le i \le w$, let $X_i$ be a random variable indicating the number of false positives of $A_i$ in $U_i$. Let $X=\sum_{i=1}^{w}X_i$ be the total number of false positives of $A'$. We bound the probability that $X$ is too high.
\begin{claim}
For $w=\frac{\logp{6n'}}{\varepsilon^2}$ we have that $\Prp{X > 3\varepsilon u'} \le \frac{1}{6n'}$
\end{claim}
\begin{proof}
Define $Y_i=\frac{1}{u}\sum\limits_{j=1}^{i}\parensqr{X_i-\E[X_i]}$ and $Y_0=0$. Since $\E[X_i]\le \varepsilon u$ we get that
$$\Prp{X > 3\varepsilon u'} = \Prp{\sum_{i=1}^{w}X_i - \sum_{i=1}^{w}\E[X_i] > 3\varepsilon u' - \sum_{i=1}^{w}\E[X_i]} \le \Prp{Y_w > 2\varepsilon w}
$$
To bound $Y_w$, we use Azuma's inequality (in the form of \cite[Theorem 7.2.1]{AlonS04}). First, note that $\{Y_i:i=1,\dots,w\}$ is a martingale:
$$
\Ep{Y_{i+1}-Y_i|Y_1,\dots,Y_i} = \frac{1}{u}\Ep{X_{i+1}-\E[X_i]|Y_1,\dots,Y_i}=\frac{1}{u}\Ep{X_i - \E[X_i]}=0.
$$
Moreover, we have $|Y_{i+1}-Y_i|=|\frac{1}{u}\paren{X_{i+1}-\E[X_i]}| \le 1$. Thus, by Azuma's inequality we get

$$
\Prp{X > 3\varepsilon u'} \le \Prp{Y_w > 2\varepsilon w} \le e^{-4\varepsilon^2 w} \le \frac{1}{6n'}
$$

which holds for $w \ge \frac{\logp{6n'}}{\epsilon^2}$.
\end{proof}

\begin{claim}
The false positive assumption holds for $A'$ for valid sequences. Namely, for any sequence $\sigma \in W$ it holds that
$$
\Pr[\exists i \le 3n': |\{x \in U: A'(\sigma_i,x)=`Yes'\}| \ge n' + 3\varepsilon u'] \le \frac{1}{2}
$$
\end{claim}
\begin{proof}
By the previous claim, we know that for any $i$: $\Pr[|\{x \in U: A'(\sigma_i,x)=`Yes'\}| \ge n' + 3\varepsilon u'] \le \frac{1}{6n'}$. Using a union bound we get that
$$
\Pr[\exists i \le 3n': |\{x \in U: A'(\sigma_i,x)=`Yes'\}| \ge n' + 3\varepsilon u'] \le 3n' \cdot \frac{1}{6n'} = \frac{1}{2}.
$$
\end{proof}

We have shown that the false positive assumption holds for all sequences in $W$. To apply \Cref{corollary1} we are left to show that the entropy lower bound is large enough, namely:
\begin{claim}
$\log{|W|}= 2n'\log{u'} - O(n')$
\end{claim}
\begin{proof}
We count the number of possible sequences in $W$ of length $2n'$. First we need to choose $2n'/w$ elements from each of the $w$ worlds. Then, for each $2n'/w$ chosen elements, we divide them between the $2n'/w$ chucks of the sequence, and finally we count all possible orderings of each chuck. Altogether we get:
$$
{u \choose \frac{2n'}{w}}^w \cdot \paren{\frac{2n'}{w}!}^w \cdot \paren{w!}^{\frac{2n'}{w}}
$$

The entropy lower bound is:
$$
2n'\logp{u'/2n'} + O(n') + 2n'\logp{2n'/w} - 2n' + o(n') + 2n'\log{w} - 2n' + o(n') = $$
$$
 2n'\log{u'} - O(n')
$$
\end{proof}

Let $|A'|$ be the memory consumption of $A'$. We can now apply \Cref{corollary1} on $A'$ with the set of sequences $W$ and parameters $n',m',u'$ and get
$$
|A'| \ge n'\leps + n' \cdot \maxparen{\log{\frac{n'}{m'}},\log{\leps}} - O(n').
$$
Since $|A'|=\sum_{i=1}^{w}|A_i|$ we get that there exist some $i$ such that
$$
|A_i| \ge n'/w\leps + n'/w \cdot \maxparen{\log{\frac{n'/w}{m'/w}},\log{\leps}} - O(n'/w) = 
$$
$$
n\leps + n \cdot \maxparen{\log{\frac{n}{m}},\log{\leps}} - O(n) 
$$
Since $A_i$ is an $(n,m,\varepsilon)$-Sliding Bloom Filter we get the desired lower bound.
\end{proof}

\section{Acknowledgments}
We thank Ilan Komargodski, Tal Wagner and the anonymous referees for many useful comments.

\bibliographystyle{amsalpha}
\bibliography{SlidingBloomFilter}

\end{document}